\RequirePackage{fix-cm}
\documentclass[smallextended]{svjour3}       
\smartqed  

\usepackage{latexsym,amssymb}
\usepackage{amsmath}
\usepackage{epsfig,verbatim}
\usepackage{color}

\usepackage[latin1]{inputenc}
\usepackage[T1]{fontenc}

\newtheorem{defi}{Definition}[section]

\newtheorem{lem}[defi]{Lemma}

\newcommand{\be}{\begin{equation}}
\newcommand{\ee}{\end{equation}}

\newcommand{\R}{\mathbb{R}}
\newcommand{\liq}{L^\infty(Q)}
\newcommand{\ds}{\displaystyle}
\newcommand{\om}{\Omega}

\newfont{\deut}{eufm10}
\newcommand{\vsp}{\vspace{1ex}}

\journalname{Computational Optimization and Applications}
\begin{document}

\title{Optimization of nonlocal time-delayed feedback controllers
\thanks{This work was supported by DFG in the framework of the Collaborative
Research Center SFB 910, projects A1 and B6.}}

\titlerunning{Nonlocal time-delayed feedback}        

\author{P. Nestler \and E.  Sch\"oll \and  F. Tr\"oltzsch}

\institute{Peter Nestler \at
Institut f\"ur Mathematik,
Technische Universit\"at Berlin, D-10623 Berlin, Germany
\\
\email{nestler@math.tu-berlin.de}           
           \and
           Eckehard Sch\"oll \at
Institut f\"ur Theoretische Physik, Technische Universit\"at Berlin, D-10623 Berlin, Germany\\
\email{schoell@physik.tu-berlin.de}
 \and
 Fredi Tr\"oltzsch\at
 Institut f\"ur Mathematik,
 Technische Universit\"at Berlin, D-10623 Berlin, Germany\\
 \email{troeltzsch@math.tu-berlin.de}
 }

\date{Received: date / Accepted: date}

\maketitle

\begin{abstract}
A class of Pyragas type nonlocal feedback controllers with time-delay is investigated for the Schl\"ogl model.
The main goal is to find an optimal kernel in the controller such that the associated solution of the controlled
equation is as close as possible to a desired spatio-temporal pattern. An optimal kernel is the solution to
a nonlinear optimal control problem with the kernel taken as control function. The well-posedness of the optimal
control problem and necessary optimality conditions are discussed for different types of kernels.
Special emphasis is laid on time-periodic
functions as desired patterns. Here, the cross correlation between the state and the desired pattern
is invoked to set up an associated objective
functional that is to be minimized. Numerical examples are presented for the 1D Schl\"ogl model and
a class of simple step functions for the kernel.
\end{abstract}


%
%

\section{Introduction}
\label{S1} \setcounter{equation}{0}

In this paper, we consider a class of nonlocal feedback controllers
with application to the control of certain  nonlinear partial differential equations.
The research on feedback control laws of this type has become quite active in theoretical physics
for stabilizing  wave-type solutions of reaction-diffusion systems such as the Schl\"ogl model
(also known as Nagumo or  Chafee-Infante equation) or the FitzHugh-Nagumo system.

The controllers can be characterized as follows: First of all, they are a generalization of Pyragas
type controllers that became very popular in the past. We refer to \cite{pyragas1992}, \cite{pyragas2006},  
 and the survey volume \cite{schoell_schuster2008}. In the simplest form of Pyragas type feedback control, 
the  difference of the  current
state $u(x,t)$ and the retarded state $u(x,t-\tau)$, multiplied with a real number $\kappa$,
is taken as control, i.e. the feedback control $f$ is
\[
f(x,t) :=  \kappa \, (u(x,t) - u(x,t - \tau)),
\]
where $\tau$ is a fixed time delay and $\kappa$ is the feedback gain.

In the nonlocal generalization we consider in this paper, the feedback control is set up  by an integral operator of the form
\be
 f(x,t) := \kappa \, \left(\int_{0}^{T} g(\tau) u(x,t-\tau)\,d\tau - u(x,t)\right).
\label{E1.1}
\ee
Here, different time delays appear in a distributed way. Depending on the particular choice of the kernel $g$,
various spatio-temporal patterns of the controlled solution $u$ can be achieved. We refer to \cite{bachmair_schoell14,loeber_etal2014,siebert_schoell14}, and \cite{siebert_alonso_baer_schoell14} with application to the Schl\"ogl model and to \cite{atay2003,kyrichenko_blyuss_schoell2014,wille_lehnert_schoell2014} with respect to control of ordinary differential 
equations.

Our main goal  is the selection of the kernel $g$ in an optimal way. We want to achieve a desired
spatio-temporal pattern for the resulting state function and look for an optimal feedback kernel $g$ to approximate this
pattern as closely as possible.
For this purpose, in the second half of the paper we will concentrate on a particular  choice of $g$ as
a step function.

We are optimizing feedback controllers but we shall apply methods of optimal control to
achieve our goal. This leads to new optimal control problems for reaction-diffusion equations
containing nonlocal terms with time delay in the state equation.
We develop the associated necessary optimality conditions and discuss numerical approaches
for solving the problems posed. Working on this class of problems, we observed that standard quadratic tracking
type objective functionals are possibly not the right tool for approximating desired time-periodic patterns. We found
out that  the so-called cross correlation partially better fits to our goals. We report on our numerical tests at the end
of this paper.

This research contributes results to the optimal control of nonlinear reaction diffusion
equations, where wave type solutions such as traveling wave fronts or spiral waves occur in unbounded domains.
We mention the papers \cite{borzi_griesse06,brandao_etal08,ckp09} on the optimal control of systems that develop spiral waves or \cite{kun_nag_cha_wag2011,kunisch_wagner2012-3,kunisch_wang12} on systems with heart medicine as
background. Moreover, we refer to \cite{buch_eng_kamm_tro2013,casas_ryll_troeltzsch2014b},
where different numerical and theoretical aspects of optimal control of the Schl\"ogl or FitzHugh-Nagumo equations are
discussed. It is a characteristic feature of such systems that the computed optimal solutions might be unstable
with respect to perturbations in the data, in particular initial data.

Feedback control aims at generating stable solutions. Various techniques of feedback control are known, we refer
only to the monographies \cite{coron07,lasiecka_triggiani2000a,lasiecka_triggiani2000b,Krstic2010} and to the references cited therein. Moreover, we mention \cite{gugat_troeltzsch2013} on feedback stabilization for the Schl\"ogl model. Pyragas type feedback control is one
associated field of research that became very active, cf. \cite{schoell_schuster2008} for an account on current research in this field. In associated publications, the feedback control laws were considered as given. For instance, the kernel in nonlocal delayed
feedback was given and it was studied what kind of patterns arise from different choices of the kernel.

The novelty of our paper is that we study an associated inverse (say design) problem: Find a kernel such that the
associated feedback solution best approximates a desired pattern.

\section{Two models of feedback control}
\label{S2} \setcounter{equation}{0}

We consider the following semilinear parabolic equation with reaction term $R$ and control function
(forcing) $f$,
\begin{equation}
 \partial_t u - \Delta u +R(u) = f
\label{E:2.1}
\end{equation}
subject to appropriate initial and boundary conditions in a spatio-temporal domain $\Omega \times (0,T)$. Using a feedback control  in the form
\eqref{E1.1},
we arrive at the following nonlinear initial-boundary value problem that includes a nonlocal
term with time delay,
\be  \label{E:2.2}
\begin{array}{rcll}
\partial_tu(x,t) - \Delta u(x,t)+R(u(x,t)) &=&\ds \kappa \left( \int_0^T g(\tau) u(x,t-\tau)\,d\tau - u(x,t)\right)& \mbox{in }\om, \\[2ex]
 u(x,s)&=& u_0(x,s)& \mbox{in }\om,\\[1ex]
\ds {\partial_n u}(x,t) &=& 0&\mbox{on } \Gamma,
\end{array}
\ee
for almost all $t \in (0,T), \, s \in [-T,0]$.
\vsp

Here, $\partial_n$ denotes the outward normal derivative on $\Gamma=\partial \Omega$.
We want to determine a feedback kernel $g \in L^\infty(0,T) $ such that the solution $u$ to \eqref{E:2.2}
is as close as possible to a desired function $u_d$. The function $g$ will have to obey certain restrictions, namely
\begin{eqnarray}
0 \le g(t) &\le& \beta \quad \mbox{ a.e.  on } [0,T], \label{E:2.3}\\
\ds \int_0^T g(s)\, ds &=& 1,\label{E:2.4}
\end{eqnarray}
where $\beta> 0$ is a given (large) positive constant. This upper bound is chosen to have a uniform bound for $g$.
It is needed for proving the solvability of the optimal control problem.

We shall present the main part of our theory for the general type of $g$ defined above. In our numerical computations, however, we will concentrate on functions $g$ of the following particular form: We select $t_1,\, t_2$ such that $0 \leq t_1 < t_2 \leq T$, $t_2-t_1\ge \delta > 0$
 and define
\begin{equation} \label{E:2.5}
 g(t) = \left\{
 \begin{array}{cl} \ds \frac{1}{t_2 -t_1}, & t_1 \leq t \leq t_2 \\[2ex]
 0, & \mbox{elsewhere}.
 \end{array} \right.
\end{equation}
It is obvious that $g$ satisfies the constraints \eqref{E:2.3},\eqref{E:2.4} with $\beta = 1/\delta$.
Using this form for $g$, we end up with the particular feedback equation
\be \label{E:2.6}
 \partial_tu(x,t) - \Delta u(x,t) +R(u(x,t)) = \kappa \, \left(\frac{1}{t_2-t_1} \int_{t_1}^{t_2} u(x,t-\tau) \, d\tau - u(x,t)\right).
\ee
In \eqref{E:2.6}, we will also vary $\kappa$ in the state equation as part of the control variables to be optimized. In contrast to this,
$\kappa$ is assumed to be fixed in the model with a general control function $g$.
In the special model, we have a restricted
flexibility in the optimization, because only the real numbers $k, \, t_1, \, t_2$ can be varied. Yet, we are able to generate a class of interesting time-periodic patterns.

Throughout the paper we will rely on the following
\vsp

\noindent {\bf Assumptions.} The set $\om \subset \mathbb{R}^N$, $N\leq 3$,  is a bounded Lipschitz domain;
for $N = 1$, we set $\Omega=(a,b)$. By $T > 0$, a finite terminal time is fixed. In theoretical physics, also the
choice $T = \infty$ is of interest. However, we do not investigate the associated analysis, because an infinite
time interval requires the use of more complicated function spaces. Moreover, the restriction to a bounded interval
fits better to the numerical computations. Throughout the paper, we use the notation $Q:= \Omega \times (0,T)$ and
$\Sigma = \Gamma \times (0,T)$.
for the space-time cylinder.
\begin{remark} We will often use the term ''wave type solution'' or ''traveling wave''. 
 This is a function $(x,t) \mapsto u(x,t)$ that 
 can be represented in the form $u(x,t) = v(x-c\,t)$ with some other smooth function $v$. Here, $c$ is the velocity of the
 wave type solution. Such solutions are known to exist  in $\Omega = \mathbb{R}$ but  not in
 in a bounded interval $\Omega=(a,b)$. 
 
 In our paper, the terms '' wave type solution'' or ''traveling wave''  stand for solutions of the Schl\"ogl model
 in the bounded domain $(a,b)$. We use these terms, since the computed solutions exhibit a similar behavior 
 as associated solutions in $\Omega = \mathbb{R}$.
\end{remark}

The reaction term $R$ is defined by
\be \label{E:2.7}
R(u)=\rho\, (u-u_1)(u-u_2)(u-u_3),
\ee
where $u_1\le u_2\le u_3$ and $\rho > 0$ are fixed real numbers. In our computational examples, we will take $\rho:=1$.
The numbers $u_i$, $i = 1,\ldots, 3$,  define the fixed points of the (uncontrolled) Schl\"ogl
model \eqref{E:2.1}. In view of the time delay, we have to provide initial values $u_0$ for $u$ in the interval
$[-T,0]$ for the general model \eqref{E:2.2} and in $[-t_2,0]$ for the special model \eqref{E:2.6}.
We assume $u_0 \in C(\bar \om \times[-T,0])$ or $u_0 \in C(\bar \om \times[-t_2,0])$, respectively. The desired state 
$u_d$ is assumed to be bounded and measurable on $Q$. 

\section{Well-posedness of the feedback equation}

In this section, we prove the existence and uniqueness of a solution to the general feedback equation \eqref{E:2.2}.
To this aim, we first reduce the equation to an inhomogeneous initial-boundary value problem. For $t \in [0,T]$, we write
\begin{eqnarray*}
  \int_0^T g(\tau) u(x,t-\tau) \, d\tau &=&  \int_0^t g(\tau) u(x,t-\tau) \, d\tau +\underbrace{\int_t^T g(\tau) u(x,t-\tau) \, d\tau }_{=: U_g(x,t)} \\
&=& \int_0^t g(\tau) u(x,t-\tau) \, d\tau +  U_g(x,t).
\end{eqnarray*}
The function $U_g$ is associated with the fixed initial function $u_0$ and is defined by
\[
U_g(x,t) = \int_t^T  g(\tau) u_0(x,t-\tau) \, d\tau;
\]
notice that we have $t-\tau \le 0$ in the integral above. By the assumed continuity of $u_0$, the function $U_g$
belongs to $C(\bar \om \times [0,T])$.

Next, for given $g \in L^2(0,T)$, we introduce a linear integral operator $K(g): L^2(Q) \to L^2(Q)$ by
\be \label{E:3.1}
(K(g)u)(x,t):= \int_0^t g(\tau) u(x,t-\tau) \, d\tau.
\ee
Substituting $s = t-\tau$, we obtain the equivalent representation
\[
(K(g)u)(x,t)=   \int_0^t  g(t-s) u(x,s) \, ds.
\]
Inserting $U_g$ and $K(g)$ in the state equation \eqref{E:2.2}, we obtain the following nonlocal
initial-boundary value problem:
\be \label{E:3.2}
 \left\{ \begin{array}{rcll} \partial_tu -\Delta u +R(u) + \kappa \, u -  \kappa \,  K(g) u  &=&  \kappa \, U_g& \mbox{in }Q, \\[1ex]
          u(x,0)&=&u_0(x,0)& \mbox{in }\om, \\[1ex]
          \partial_n u &=&0& \mbox{on }\Sigma.
         \end{array}
\right.
\ee
In the next theorem, we use the Sobolev space
\[
W(0,T) = L^2(0,T;H^1(\om)) \cap H^1(0,T;L^2(\om)).
\]
\begin{theorem} \label{T:3.1}
 For all $g \in L^\infty(0,T)$, $U_g \in L^p(Q), \, p > \frac{5}{2},$ and $u_0 \in C(\bar \om \times [-T,0])$,
 the problem \eqref{E:3.2} has a unique solution $u \in W(0,T) \cap C(\bar{Q})$.
\end{theorem}
\begin{proof} We use the same technique that was applied in \cite{casas_ryll_troeltzsch2014} to show the
existence and continuity of the solution to the FitzHugh-Nagumo system. Let us mention the main steps.
First, we apply a simple transformation that is well-known in the theory of
evolution equations. We set
 \[
  u = e^{\lambda t} v
 \]
 with some $\lambda >0$. This transforms the partial differential equation in \eqref{E:3.2} to an equation
 for the new unknown function $v$,
 \begin{equation} \label{E:3.3}
 v_t - \Delta v + e^{-\lambda t}R(e^{\lambda t} v) + (\lambda  + \kappa) v =  \kappa \, K_\lambda(g) v  +
 e^{-\lambda t} \kappa \, U_g,
 \end{equation}
 where the integral operator $K_\lambda(g)$ is defined by
 \[
(K_\lambda(g) v)(x,t) = \int_0^t e^{-\lambda(t-s) } g(t-s)v(x,s)\, ds.
 \]
 If $g \in L^\infty(0,T)$, then both operators $K(g)$ and $K_\lambda(g)$ are continuous linear operators in $L^p(Q)$, for all $p\ge 1$.
 Moreover, due to the factor $e^{-\lambda(t-s) }$, the norm of $K_\lambda(g): L^2(Q) \to L^2(Q) $
 tends to zero as $\lambda \to \infty$. We obtain
 \begin{equation} \label{E:3.4}
 \|K_\lambda(g)\|_{\mathcal{L}(L^2(Q))} \le \frac{c}{\sqrt{\lambda}} \|g\|_{L^\infty(0,T)}
 \end{equation}
 with some constant $c > 0$. To have this estimate, we assumed in \eqref{E:2.3} that $g$ is uniformly bounded by the
 constant $\beta$. If $\lambda$ is sufficiently large, then we have
 \[
 \int_Q [e^{-\lambda t}R(e^{\lambda t} v) + (\lambda  + \kappa) v - \kappa K_\lambda(g)v] \, v \, dxdt \ge
 \frac{\lambda}{2} \|v\|^2_{L^2(Q)} \quad \forall v \in L^2(Q),
 \]
 because the coercive term $(\lambda + \kappa)\, v$ in the left side is dominating the other terms, cf. \cite{casas_ryll_troeltzsch2014}.

 With this inequality, an a priori estimate can be derived in $L^2(Q)$ for any solution $v$ of the equation \eqref{E:3.2}.
 Now, we can proceed as in \cite{casas_ryll_troeltzsch2014}: A fixed-point principle is applied in $L^2(Q)$ to prove
 the existence and uniqueness of the solution $v$ that in turn implies the same for $u$.
 For the details, the reader is referred to \cite{casas_ryll_troeltzsch2014}, proof
 of Theorem 2.1. However, we mention one important idea: Thanks to \eqref{E:3.4}, the term $(\lambda + \kappa)$
 absorbes the non-monotone terms in the equation \eqref{E:3.3} so that, in estimations, equation \eqref{E:3.3} behaves
 like the parabolic equation
 \[
 v_t - \Delta v + \tilde R(v) = F
 \]
 with a monotone non-decreasing nonlinearity $\tilde R$ and given right-hand side $F \in L^p(Q)$, $p > 5/2$. This fact
 can be exploited to verify, for each $r > 0$,  the existence of a constant $C_r > 0$ with the following property:
 If $g \in L^\infty(Q)$ obeys $\|g\|_{ L^\infty(Q)} \le r$ and $u$ is the associated solution to \eqref{E:2.2}, then
 \begin{equation} \label{E:3.5}
 \|u\|_{L^\infty(Q)} \le C_r.
 \end{equation}
 \hfill $\Box$ \end{proof}

\section{Analysis of optimization problems for feedback controllers}
\subsection{Definition of two optimization problems}
\subsubsection*{General kernel as control}

Let a desired function $u_d \in L^\infty(Q)$ be given. In our later applications, $u_d$ models a desired spatio-temporal
pattern. Moreover, we fix a non-negative function $c_Q \in L^\infty(Q)$. This function is used for selecting a desired observation domain.
We consider the feedback equation \eqref{E:2.2} and want to find a kernel $g$ such that the associated
solution $u$ approximates $u_d$ as close as possible in the domain of observation. This goal is expressed by the following functional
$j: L^2(Q) \times L^\infty(0,T) \to \mathbb{R}$ that is
to be minimized,
\[
j(u,g) :=\frac{1}{2} \iint_Q c_Q(u-u_d)^2 \, dxdt + \frac{\nu}{2} \int_0^T g^2(t) \, dt.
\]
Here, $\nu \ge 0$ is a Tikhonov regularization parameter. The standard choice of $c_Q$ is $c_Q(x,t) = 1$ for all $(x,t) \in Q$.
Another selection will be applied for  periodic functions $u_d$: $c(x,t) = 1$ for all $(x,t)  \in Q$ with $t \ge T/2$ and $c(x,t) = 0$ 
for all $(x,t) \in Q$ with $t <  T/2$.

By Theorem \ref{T:3.1}, to each $g \in L^\infty(0,T)$ there exists a unique associated state function $u$
that will be denoted by $u_g$. Then $j$ does only depend on $g$ and we obtain the reduced objective functional $J$,
\[
 J: g \mapsto  j(u_g,g).
\]
Therefore, our general optimization problem can be formulated as follows:
\[
\tag{PG}  \min_{g \in C} J(g):= \frac{1}{2} \iint_Q c_Q(u_g-u_d)^2 \, dxdt + \frac{\nu}{2} \int_0^T g^2(t) \, dt,
\]
where $C \subset L^\infty(0,T)$ is the convex and closed set defined by
\[
C:= \left\{ g \in L^\infty(0,T): \ 0 \le g(t)  \le  \beta \; \mbox{ a.e. in } [0,T] \mbox{ and }  \int_0^T g(t) \, dt =1.\right\}
\]
Notice that $C$ is a weakly compact subset of $L^2(0,T)$. The restrictions on $g$ are motivated by the
background in mathematical physics. In particular, the restriction on the integral of $g$ guarantees that
\[
\int_0^T g(\tau)u(x,t-\tau)\, d\tau - u(x,t) = 0,
\]
if $u(x,t-\tau)= u(x,t)$ in Q.
By the definition of $u_g$, the optimization is  subject to the state equation \eqref{E:2.2}.

\subsubsection*{Special kernel as control}
The other optimization problem we are interested in, uses the particular form \eqref{E:2.5} of the kernel $g$,
\[
 \min_{0 \leq t_1 < t_2 \leq T} J_S(\kappa,t_1,t_2):= \frac{1}{2} \iint_Q c_Q(u_{(\kappa,t_1,t_2)}-u_d)^2 \, dxdt + \frac{\nu}{2} (t_1^2+t_2^2 + \kappa^2),
\]
where $u_{(\kappa,t_1,t_2)}$ is the solution of \eqref{E:2.6} for a given triplet $(\kappa,t_1,t_2)$.
This problem might fail to have an optimal solution, because the set of admissible triplets
$(\kappa,t_1,t_2)$ is not closed. Notice that we need $t_1<t_2$ in
\eqref{E:2.6}. Therefore, we fix $\delta > 0$ and define the slightly changed admissible set
\[
C_\delta:= \left\{ (\kappa,t_1,t_2) \in \R^3: \ 0 \leq t_1 < t_2 \leq T, \, t_2 -t_1 \ge \delta, \  \kappa \in \R \right\}
\]
that is compact. In this way, we obtain the special finite-dimensional optimization problem for step functions $g$,
\[
\tag{PS} \min_{(\kappa,t_1,t_2) \in C_\delta} J_S(\kappa,t_1,t_2):= \frac{1}{2} \iint_Q c_Q(u_{(\kappa,t_1,t_2)}-u_d)^2 \, dxdt + \frac{\nu}{2} (t_1^2+t_2^2+ \kappa^2).
\]

\subsection{Discussion of (PG)}
\subsubsection*{The control-to-state mapping $G$}

Next, we discuss the differentiability of the  control-to-state mappings  $g \mapsto u_g$ and $(\kappa,t_1,t_2)  \mapsto
u(\kappa,t_1,t_2)$.  First, we consider the case of the general kernel $g$.
The analysis for the particular kernel \eqref{E:2.5} is fairly analogous but cannot deduced as a particular case
of (PG). We will briefly sketch it in a separate section.

By Theorem \ref{T:3.1}, we know that
the mapping $G: g \mapsto u_g$ is well defined from $L^\infty(0,T)$ to $C(\bar Q)$.
Now we discuss the differentiability of $G$.
To slightly simplify the notation, we introduce an operator
$\mathcal{K}: L^\infty(0,T)\times C(\bar Q) \to C(\bar Q)$ by
\[
\mathcal{K}(g,u) = K(g)u,
\]
where $K(g)$ was introduced in \eqref{E:3.1}; notice that $\mathcal{K}$ is bilinear.
Let us first show the differentiability for  $\mathcal{K}$.

We fix $g \in L^\infty(0,T), \, u \in C(\bar Q)$, and select varying increments $h \in L^\infty(0,T)$, $v \in C(\bar Q)$.
Then we have
\[
\begin{split}
&\mathcal{K}(g+h,u+v)= \int_0^T [g(\tau)+h(\tau)] [u(x,t-\tau)+v(x,t-\tau)] \, d\tau \\
&\quad =   \int_0^t g(\tau)u(x,t-\tau) \, d\tau +
\underbrace{\int_0^t h(\tau)u(x,t-\tau) \, d\tau +   \int_0^t g(\tau)v(x,t-\tau)\, d\tau}_{A(g,u)(h,v)}\\
&\quad \qquad + \underbrace{\int_0^t h(\tau)v(x,t-\tau)\, d\tau}_ {R(h,v)}
 = {\mathcal{K}(g,u)} + A(g,u)(h,v)+ {R(h,v)},
\end{split}
\]
where $A(g,u): L^\infty(0,T) \times C(\bar Q) \to C(\bar Q)$ is a linear continuous operator and
$R: L^\infty(0,T) \times C(\bar Q) \to C(\bar Q)$ is a remainder term. It is easy to confirm that
\[
\frac{\|R(h,v) \|_{C(\bar Q)}}{\|(h,v) \|_{L^\infty(0,T)\times C(\bar Q)}}  \to 0, \quad  \mbox{ if } \|(h,v) \|_{L^\infty(0,T)\times C(\bar Q)} \to 0.
\]
Therefore, $\mathcal{K}$ is Fr\'{e}chet-differentiable. As a continuous bilinear form, $\mathcal{K}$ is also of
class $C^2$.

Now, we investigate the control-to-state mapping $G: L^\infty(0,T) \to C(\bar Q)$ defined by
$
 G: g \mapsto u_g,
$
where the state function $u_g$ is defined as the unique solution to
\be \label{E:4.1}
\begin{array}{rcll}
 \partial_tu - \Delta u + R(u) +\kappa \, u &=&  \kappa\,\mathcal{K}(g,u)+ \kappa\,U_g & \text{ in } Q\\
  \partial_n u &=& 0&\text{ in } \Sigma\\
 u(0)&=& u_0(0)&\text{ in }\Omega .
 \end{array}
\ee
In what follows, the initial function $u_0$ will be kept fixed
and is therefore not mentioned. Of course, $U_g,\, G$ and some of the operators below depend on $u_0$, but we will
not explicitely mention this dependence.
To discuss $G$, we need known properties
of the following auxiliary mapping $\mathcal{G}: v \mapsto u$, where
\[
\begin{array}{rcll}
 \partial_tu - \Delta u + R(u) +\kappa \, u &=& v & \text{ in } Q\\
 \partial_n u &=& 0& \text{ in } \Sigma\\
 u(0)&=& u_0(0)& \text{ in } \Omega.
\end{array}
\]
This mapping $\mathcal{G}$ is of class $\mathcal{C}^2$ from $L^p(Q)$ to $W(0,T) \cap \mathcal{C}(\bar{Q})$,
if $p>\frac{5}{2}$, in particular from $L^\infty(Q)$ to $L^\infty(Q)$, cf. \cite{casas_ryll_troeltzsch2014} or, for monotone
$R$, \cite{cas93}, \cite{rayzid99}, \cite{tro10book}.
\smallskip

Now (consider $v:= \kappa\,( \mathcal{K}(g,u)+U_g)$ as given and keep the initial function $u_0$ fixed),
$u$ solves \eqref{E:4.1} if and only if  $u=\mathcal{G}( \kappa\,\mathcal{K}(g,u)+\kappa\, U_g)$, i.e.
\be\label{E:4.2}
 u-\mathcal{G}( \kappa\,\mathcal{K}(g,u)+\kappa\, U_g)=0.
\ee
We introduce a  new mapping
$\mathcal{F}\,:\,L^\infty(Q)\times L^\infty(0,T)  \to L^\infty(Q)$ defined by
\[
 \mathcal{F}(u,g):=u-\mathcal{G}( \kappa\,\mathcal{K}(g,u)+\kappa\, U_g).
\]
Then, \eqref{E:4.2} is equivalent to the equation
\be \label{E:4.3}
 \mathcal{F}(u,g)=0.
\ee
We have proved above  that the mapping
 $(g,u) \mapsto \mathcal{K}(g,u)$ is of  class $\mathcal{C}^2$ from $L^\infty(0,T) \times L^\infty(Q)$ to $L^\infty(Q)$.
 Obviously, also the linear mapping $g \mapsto U_g$ is of class $\mathcal{C}^2$ from $L^\infty(0,T)$ to $L^\infty(Q)$.
 By the chain rule, also $\mathcal{F}$  is  $\mathcal{C}^2$  from $L^\infty(Q)\times L^\infty(0,T)  \to L^\infty(Q)$
 and the mappings $\partial_g \mathcal{F}(\bar{u},\bar{g})$,
  $\partial_u \mathcal{F}(\bar{u},\bar{g})$ are continuous in the associated pairs of spaces.

To use the implicit function theorem, we prove that
$\partial_u \mathcal{F}(\bar{u},\bar{g})$ is continuously invertible at any fixed pair $(\bar{u},\bar{g})$.  Therefore, we
consider the equation
\be\label{E:4.4}
 \partial_u \mathcal{F}(\bar{u}, \bar{g})v=z
\ee
with given right-hand side $z \in L^\infty(Q)$ and  show the existence of a unique solution $v \in L^\infty(Q)$. The equation is equivalent with
\be
 v- \mathcal{G}'(\underbrace{\kappa \, \mathcal{K}(\bar g,\bar u)+\kappa \, U_g}_{\bar{p}}) \kappa\, \mathcal{K}(g,v)=   z.
 \ee
Writing for convenience $\bar p = \kappa\,\mathcal{K}(\bar g,\bar u)+ \kappa\,U_g$, we obtain the simpler form
 \[
 v- \mathcal{G}'(\bar{p})\kappa\,K(\bar{g})v= z.
 \]
A function $z \in \liq$ does not in general belong to $W(0,T)$. To overcome this difficulty,  we set $w:=v-z$ and transform
the equation to
\be \label{E:4.6}
 w=\mathcal{G}'(\bar{p})\underbrace{\kappa\,K(\bar{g})(w+z)}_{q}=
 \mathcal{G}'(\bar{p})q.
 \ee
 where $q := \kappa\,K(\bar{g})(w+z)$.
As the next result shows, $w$ is the solution of a parabolic PDE, hence $w \in W(0,T)$.
\begin{lem} \label{L5.1}
 Let $q \in L^p(Q)$ with $p > 5/2$ be given. Then we have $y=\mathcal{G}'(\bar{p})q$ if and only if $y$ solves
\[
\begin{array}{rcll}
 \partial_t y - \Delta y +R'(\bar{u}) y+ \kappa\,y &=&q& \text{ in } Q\\
 \partial_n y &=& 0& \text{ in } \Sigma\\
 y(0)&=&0& \text{ in } \Omega,
 \end{array}
 \]
 where $\bar{u}$ is the solution  associated with  $\bar{p}$, i.e.
\[
\begin{array}{rcll}
\partial_t \bar{u} - \Delta \bar{u} +R(\bar{u})+ \kappa\,\bar{u} &=&\bar{p}& \text{ in }  Q\\
  \partial_n \bar u &=& 0& \text{ in }\Sigma \\
 \bar{u}(0)&=&u_0& \text{ in } \Omega.
 \end{array}
 \]
\end{lem}
We refer to \cite{casas_ryll_troeltzsch2014}. For monotone non-decreasing functions $R$, this result  is well  known  in the theory of semilinear parabolic control problems, see e.g. \cite{cas93}, \cite{rayzid98}, or \cite[Thm. 5.9]{tro10book}.
By Lemma \ref{L5.1},  the solution  $w$ of \eqref{E:4.6}  is the unique solution of the linear PDE
\be \label{E:4.7}
\begin{split}
& (\partial_t w- \Delta w +R'(\bar{u})w+ \kappa\,w)(x,t)= q(x,t)\\
 &  \qquad \qquad =  \kappa\int_0^t \bar{g}(\tau) w(x,t-\tau) \, d\tau + \kappa\int_0^t \bar{g}(\tau) z(x,t-\tau) \, d\tau
 \end{split}
 \ee
 subject to $w(0) = 0$ and homogeneous Neumann boundary conditions.
 By the same methods as above we find that,  for all $z \in L^\infty(Q)$, equation $(\ref{E:4.7})$ has
 a unique solution $w \in W(0,T) \cap L^\infty(Q)$. \smallskip

After transforming back by $v = w + z$, we have found that for all $z \in L^\infty(Q)$, \eqref{E:4.4}  has  a unique solution $v \in L^\infty(Q)$ given by $v = w + z$. Therefore, the inverse operator $\partial_u \mathcal{F}(\bar{u},\bar{g})^{-1}$ exists.
 The  continuity of this inverse mapping follows from a
 result of  \cite{casas_ryll_troeltzsch2014} that the mapping $z \mapsto w$ defined by \eqref{E:4.7}
 is continuous in $L^\infty(Q)$. 

Next, we consider the operator $\partial_g \mathcal{F}$.
It exists by the chain rule and admits the form
\[
 \partial_g \mathcal{F}(\bar{u},\bar{g})h=\mathcal{G}'(\kappa\,(K(\bar{g})\bar{u}+U_g))\kappa\,(K(h)\bar{u}+ \partial_gU_g h).
\]
Setting again $\bar p = \kappa\,(K(\bar{g})\bar{u}+U_g)$ and $q = \kappa\,(K(h)\bar{u}+ \partial_gU_g h)$, 
we see that 
\[
\partial_g \mathcal{F}(\bar{u},\bar{g})h=\eta,
\]
where, by Lemma \ref{L5.1},  $\eta$ solves the equation 
\[
\partial_t \eta - \Delta \eta + R'(\bar{u})\eta+\kappa\,\eta= q = \kappa\,K(h)\bar{u}+ \kappa\, \partial_gU_g h
\]
subject to homogeneous initial and boundary conditions.
 Therefore, $\eta$ is the unique solution to
\begin{eqnarray*}
( \partial_t \eta - \Delta \eta + R'(\bar{u})\eta+\kappa\,\eta)(x,t)&=& 
\kappa\, \int_0^t h(\tau)\bar{u}(x,t-\tau)\, d\tau\\
&&\quad + \kappa\, \int_t^T h(\tau)u_0(x,t-\tau)\, d\tau\\
 \eta(x,0)&=&0\\
  \partial_n \eta &=& 0.
\end{eqnarray*}
By $\bar u(x,t) = u_0(x,t)$ for $-T \le t \le 0$, we can re-write this as
\begin{eqnarray*}
( \partial_t \eta - \Delta \eta + R'(\bar{u})\eta+\kappa\,\eta)(x,t)&=& \kappa\, \int_0^T h(\tau)\bar{u}(x,t-\tau)\, d\tau\\
 \partial_n \eta &=& 0\\
 \eta(x,0)&=&0.
\end{eqnarray*}
Again, the mapping  $h \mapsto w$ is continuous from
$L^\infty(0,T)$ to $W(0,T) \cap \mathcal{C}(\bar{Q})$.
\medskip

Collecting the last results, we have the following theorem:
\begin{theorem}[Differentiability of $G$] The control-to-state mapping $G: g \mapsto u_g$ associated with equation \eqref{E:3.2} is
of class $C^2$. The first order derivative $z:=G'(g)h$ is obtained as the unique solution to
\be\label{E:4.8}
\begin{array}{rcll}
(\partial_t z - \Delta z+R'(u_g)z+\kappa \, z)(x,t)&=& \ds \kappa \int_0^T h(\tau) u_g(x,t- \tau) \, d \tau\\
&&\ds \qquad   +\kappa \int_0^t g\, (\tau) z(x,t- \tau) \, d \tau&\text{ in }Q\\[1ex]
\partial_n z &=& 0&\text{ in } \Sigma\\
z(\cdot,t)&=& 0, \; -T \leq t \leq 0& \text{ in } \Omega.
\end{array}
\ee
\end{theorem}
\begin{proof}
We already know by Theorem \ref{T:3.1} that, for all $g\in L^\infty(0,T)$, there exists  a unique solution $u=G(g) \in W(0,T)
\cap \mathcal{C}(\bar{Q})$ solving the equation
\[
 \mathcal{F}(u,g)=0.
\]
We discussed above that the assumptions of the implicit function theorem are satisfied.
Now this theorem yields that  the mapping $g \mapsto G(g)$
is of class $\mathcal{C}^2$.

The derivative $G'(g)h$ is obtained by implicit differentiation.
By definition of $G(g)$, we have
\be \label{E:4.9}
\begin{array}{rcl}
(\partial_t G(g) - \Delta G(g)+R(G(g))+\kappa \, G(g))(x,t)&=& \ds \kappa \int_0^t g(\tau) \, G(g)(x,t- \tau) \, d \tau  \\[2ex]
&& \ \ + \ds \kappa \, \int_t^T g(\tau) \, u_0(x,t- \tau) \, d \tau\\[1ex]
\partial_n G(g) &=& 0\\[1ex]
G(g)(\cdot,t)&=& u_0(\cdot,t), \; -T \leq t \leq 0.
\end{array}
\ee
Implicit differentiation yields  that $z:= G'(g)h$ is the unique solution of \eqref{E:4.8}.
Notice that
\[
\int_0^t g(\tau) \, G(g)(x,t- \tau) \, d \tau +  \int_t^T g(\tau) \, u_0(x,t- \tau) \, d \tau = \int_0^T g(\tau) \, G(g)(x,t- \tau) \, d \tau.
\]
\hfill $\Box$ \end{proof}

\subsection{Existence of an optimal kernel}

\begin{theorem}
 For all  $\nu \ge 0$, {\rm (PG)}  has at least one optimal solution $\bar{g}$.
\end{theorem}
\begin{proof}
Let $(g_n)$ with $g_n \in C$ for all $n \in \mathbb{N}$ be a minimizing sequence. Since $C$ is
bounded, convex, and closed in $L^\infty(0,T)$, we can assume
without limitation of generality that $g_n$ converges weakly in  $L^2(0,T)$ to $\bar g$, i.e. 
$g_n \rightharpoonup \bar{g}$, $n \to \infty$. The associated sequence
of states $u_n$ obeys the equations
\begin{equation} \label{E:4.10}
\partial_t u_n - \Delta u_n + \kappa \, u_n = d_n:=
-\kappa \, R(u_n) + \kappa \, K(g_n)u_n + \kappa \, U_g.
\end{equation}
By the principle of superposition, we split the functions $u_n$ as $u_n = \hat{u} + \tilde u_n$, where $\hat{u}$ is the solution
of \eqref{E:4.10} with right-hand side $d_n := 0$ and initial value $\hat{u}(0) = u_0(0)$, while $\tilde u_n$ is the solution to the
right-hand side $d_n$ defined above and zero initial value.
In view of \eqref{E:3.5}, all state functions $u_n$, hence also the functions $\tilde u_n$, are uniformly bounded in $L^\infty(Q)$.
Thanks to  \cite[Thm. 4]{DiBenedetto1986}, the sequence $(\tilde u_n)$ is bounded in some H\"older space $C^{0,\lambda}(Q)$. By the Arzela-Ascoli
theorem, we can assume (selecting a subsequence, if necessary) that $\tilde u_n$ converges strongly in $L^\infty(Q)$.
Adding to $\tilde u_n$ the fixed function $\hat u$, we have that $(u_n)$ converges strongly to some  $\bar u$ in $L^\infty(Q)$.

The boundedness of $(u_n)$ also induces the boundedness of the sequence $(d_n)$ in $L^\infty(Q)$, in particular
in $L^2(Q)$. Therefore, we can assume that $d_n$ converges weakly in $L^2(Q)$ to $ \bar d$, $n \to \infty$. Since $(u_n)$ is the sequence of
solutions to the ''linear'' equation \eqref{E:4.10} with right-hand side $d_n$, the weak convergence of $(d_n)$ induces the
weak convergence of $u_n \rightharpoonup \bar u$ in $W(0,T)$, where $\bar u$ solves  \eqref{E:4.10} with right-hand side
$\bar d$.

Finally, we show that
\[
\bar d(t) = -\kappa R(\bar u(t)) + \kappa \, (K(\bar g)\bar u)(t) + \kappa \, U_g(t)
\]
so that $\bar u$ is the state associated with $\bar g$. Obviously, it suffices to prove
that $K(g_n)u_n $ converges weakly to $K(\bar g)\bar u$ in $L^2(Q)$. To this aim, let an arbitrary $\varphi \in L^2(Q)$ be given.
Then we have
\begin{equation} \label{E:4.11}
\begin{split}
&\iint_Q  \varphi(x,t)\left( \int_0^t g_n(\tau)u_n(x,t-\tau)\, d\tau \right)dx dt \\
&\quad = \int_0^T g_n(\tau) \left(\int_\tau^T \int_\om \varphi(x,t)u_n(x,t-\tau) \, dt dx\right) d\tau.
\end{split}
\end{equation}
Clearly, the strong convergence of $(u_n)$ in $L^\infty(Q)$ yields
\[
\int_\tau^T \int_\om \varphi(x,t)u_n(x,t-\cdot) \, dt dx \to \int_\tau^T \int_\om \varphi(x,t)\bar u(x,t-\cdot) \, dt dx
\]
in $L^2(0,T)$. Along with the weak convergence of $g_n$, this implies
\[
\begin{split}
& \lim_{n \to \infty} \int_0^T g_n(\tau) \int_\tau^T \int_\om \varphi(x,t)u_n(x,t-\tau) \, dt dx d\tau \\
&\qquad =  \int_0^T \bar g(\tau)
 \int_\tau^T \int_\om \varphi(x,t)\bar u(x,t-\tau) \, dt dx d\tau .
 \end{split}
\]
In view of \eqref{E:4.11}, we finally arrive at
\[
\iint_Q \varphi(x,t) \int_0^t g_n(\tau)u_n(x,t-\tau)\, d\tau dx dt \to
\iint_Q\varphi(x,t) \int_0^t \bar g (\tau)\bar u(x,t-\tau)\, d\tau dx dt
\]
as $n \to \infty$. Since this holds for arbitrary $\varphi \in L^2(Q)$, this is equivalent to the desired weak convergence $K(g_n)u_n \rightharpoonup K(\bar g)\bar u$ in $L^2(Q)$.
\hfill $\Box$ \end{proof}

\subsection{Necessary optimality conditions}
\subsubsection{Adjoint equation}
In the next step of our analysis, we establish the necessary optimality conditions for a (local) solution $\bar g$ of
the optimization problem (PG).
This optimization problem is defined by 
\be \label{E:4.12}
 \left\{ \begin{array}{l} \min J(g), \\[1ex]
 \ds 0 \le g(t) \le \beta \quad \text{for almost all } t \in [0,T],\\[1ex]
  \ds \int_0^T g(\tau) \, d\tau =1. \end{array} \right.
\ee
Although the admissible set belongs to $L^\infty(0,T)$, we consider this as an optimization problem in the Hilbert space $L^2(0,T)$.

To set up associated necessary optimality conditions for an optimal solution  of \eqref{E:4.12}, we first determine a useful
expression for the derivative of the objective functional $J$. We have
\[
\begin{split}
 J(g)&=\ds \frac{1}{2} \iint_Q c_Q(u_g -u_d)^2 \, dxdt + \frac{\nu}{2} \int_0^T g(t)^2 \, dt\\
&\qquad \qquad  =
 \ds \frac{1}{2} \iint_Q c_Q (G(g) -u_d)^2 \, dxdt + \frac{\nu}{2} \int_0^T g(t)^2 \, dt.
 \end{split}
\]
Let now be an arbitrary (i.e. not necessarily optimal) $\bar g \in L^\infty(0,T)$ be given and let $\bar u = G(\bar g)$
be the associated state. Then we obtain for  $h \in L^\infty(0,T)$
\begin{eqnarray}
J'(\bar g)h &=& \nu \int_0^T  \bar g(t) \, h(t) \, dt  + \iint_Q c_Q(\bar{u}-u_d)(G'(\bar{u})h)\, dxdt \nonumber \\
&=&  \int_0^T  \nu\bar g(t) h(t) dt  + \iint_Q c_Q(x,t)(\bar{u}(x,t) -u_d(x,t)) z(x,t)dxdt \label{E:4.13b}
\end{eqnarray}
with  the solution $z$  to the equation \eqref{E:4.8} for $u_g := u_{\bar g} = \bar u$.

The implicit appearance of $h$ via $z$  can be converted to an explicit  one by an {\em adjoint equation}.
This is the following equation:
\be \label{E:4.14}
\begin{array}{rcl}
 \begin{array}{rcl} (-\partial_t \varphi - \Delta \varphi +R'(\bar{u})\varphi + \kappa\,\varphi)(x,t)
& = &  \ds \kappa\int_0^T \bar g(\tau) \varphi(x,t+\tau) \, d\tau \\[2ex]
&&\qquad + c_Q(x,t)(\bar{u}(x,t) -u_d(x,t))\\[1ex]
&& \qquad\qquad \text{ a.e. in } Q,\\[1ex]
 \partial_n \varphi&=& 0 \quad \text{in } \Sigma,\\[1ex]
   \varphi(\cdot,t)&=&0\quad t \in [T,2T].
 \end{array}
\end{array}
\ee

The solution $\bar \varphi$ of \eqref{E:4.14} is said to be the {\em adjoint state} associated with
$\bar g$.
\begin{lemma}\label{L:4.3} Let $\bar g,\, \bar h \in L^\infty(0,T)$, and $\bar u = u_{\bar g}$ be given. If $z$ is the solution to the
linearized equation \eqref{E:4.8} for $u_g := \bar u$ and $\bar \varphi$ is the unique solution to the adjoint equation \eqref{E:4.14},
then the  identity
\begin{equation} \label{E:L43}
\iint_Q (c_Q\,(\bar{u}-u_d)\, z)(x,t) \, dxdt = \kappa \, \iint_Q \bar \varphi(x,t) 
\left(\int_0^T h(\tau)\bar{u}(x,t-\tau) \, d\tau
 \right)\, dxdt
\end{equation}
 is fulfilled:
\end{lemma}
\begin{proof}
We multiply the first equation in \eqref{E:4.8} by the adjoint state $\bar \varphi$ as test function and the first equation in
\eqref{E:4.14} by $z$. After integration on $Q$ and some partial integration with respect to $x$, we obtain
\[
\begin{split}
&\iint_Q \left(\partial_t z \, \bar \varphi + \nabla z \cdot \nabla \bar \varphi + (R'(\bar u) + \kappa)z\, \bar \varphi \right)dxdt
\\
&\qquad \qquad = \kappa \, \iint_Q \left(\int_0^T \bar g(\tau) z(x,t-\tau)\, d\tau\right) 
\bar\varphi(x,t)\, dxdt \\
&\qquad \qquad \qquad
+ \kappa \iint_Q \left(\int_0^T h(\tau)\bar u(x,t-\tau)\, d\tau\right)
\bar\varphi(x,t)\, dxdt
\end{split}
\]
and
\[
\begin{split}
&\iint_Q \left(- z \, \partial_t \bar\varphi + \nabla z \cdot \nabla \bar\varphi 
+ (R'(\bar u) + \kappa)z\, \bar \varphi \right)dxdt
\\
&\qquad =  \kappa\,\iint_Q\left( \int_0^T \bar g(\tau) \bar\varphi(x,t+\tau) \, d\tau\right) z(x,t)\, dxdt
+ \iint_Q c_Q(\bar u - u_d)\, z\, dxdt.
\end{split}
\]
Integrating by parts with respect to $t$, we see that
\[
\iint_Q (- z) \, \partial_t \bar\varphi \, dxdt = \iint_Q \bar \varphi \, \partial_t z \,  dxdt;
\]
notice that we have $z(0) = 0$ and $\bar\varphi(T) = 0$. Comparing both weak formulations above, it turns out
that we only have to confirm the equation
\be\label{E:confirm}
\iint_Q \int_0^T \bar g(\tau) z(x,t-\tau)\, \bar\varphi(x,t)\, d\tau \, dxdt = \iint_Q\int_0^T \bar g(\tau)\bar \varphi(x,t+\tau) \,z(x,t)\, d\tau \, dxdt.
\ee
Then the claim of the Lemma follows.  To show \eqref{E:confirm} , we proceed as follows:
\begin{equation} \label{E:Fubini}
\begin{split}
&\iint_Q \int_0^T \bar g(\tau) z(x,t-\tau) \varphi(x,t)\, d\tau\, dxdt \hspace{4cm}\\
&\hspace{2cm} =\int_\Omega \int_0^T \int_0^T \bar g(\tau) \, z(x,t-\tau)\, \bar\varphi(x,t)\,d\tau dtdx\\
&\hspace{2cm} =\int_\Omega \int_0^T \int_0^t \bar g(\tau) \, z(x,t-\tau)\, \bar\varphi(x,t)\,d\tau dtdx\\
&\hspace{2cm}=\int_\Omega \int_0^T \int_0^t \bar g(t-\eta) \, z(x,\eta)\, \bar\varphi(x,t)\,d\eta dt\, dx\\
&\hspace{2cm}=\int_\Omega \int_0^T \int_\eta^T \bar g(t-\eta) \, \bar\varphi(x,t)\,dt \, z(x,\eta) \, d\eta\, dx\\
&\hspace{2cm}=\int_\Omega \int_{0}^T \int_0^{T-\eta} \bar g(\sigma) \, \bar\varphi(x,\eta+\sigma)\,d\sigma \, z(x,\eta) \, d\eta\, dx\\
&\hspace{2cm}=\int_\Omega \int_{0}^T \int_0^T \bar g(\sigma) \, \bar\varphi(x,\eta + \sigma)\,d\sigma \, z(x,\eta) \, d\eta\, dx\\
&\hspace{2cm}=\iint_Q \int_{0}^T \bar g(\tau) \, \bar\varphi(x,t + \tau) \, z(x,t)\,d\tau \, dxdt.
\end{split}
\end{equation}
We used $z(x,t-\tau) = 0$ for $\tau > t$ in the second equation, the substitution $\eta = t-\tau$ in the third, the Fubini theorem
in the fourth, the substitution $\sigma = t-\eta$ in the fifth, the property $\bar\varphi(x,t) = 0$ for $t \ge T$ in the sixth equation. Finally,
we re-named the variables.
\hfill $\Box$ \end{proof}

\begin{corollary} \label{Cor:4.4} At any $\bar g \in L^\infty(0,T)$, the derivative $J'(\bar g)\,h$ in the direction $h \in L^\infty(0,T)$
is given by
\[
J'(\bar g)\,h = \int_0^T \nu \, \bar g(t)\, h(t)\,dt + \kappa  \int_0^T h(\tau) \left(\iint_Q \bar \varphi(x,t) \bar{u}(x,t-\tau) \, dxdt\right) d\tau,
\]
where $\bar \varphi$ is the unique solution of the adjoint equation \eqref{E:4.14}.
\end{corollary}

This follows immediately by inserting the right-hand side of \eqref{E:L43} in \eqref{E:4.13b} and by interchanging the order
of integration with respect to $t$ and $\tau$.

\subsubsection{Necessary optimality conditions for (PG)}

Let us now establish  the necessary optimality conditions for an optimal solution $\bar g$ of \eqref{E:4.12}.
They can be derived by the Lagrangian function $L: L^\infty(0,T) \times \mathbb{R} \to \mathbb{R}$,
\[
 L(g,\mu):=J(g)+\mu\left(\int_0^T g(\tau) \, d\tau -1\right).
\]
If $\bar g$ is an optimal solution, then  there exists  a real Lagrange multiplier $\bar{\mu}$ such that
the variational inequality
\[
 J'(\bar{g})(g-\bar{g}) + \bar{\mu} \int_0^T(g-\bar{g}) \, dt \geq 0 \qquad \text{ for all } g \geq 0
\]
is satisfied. Inserting the result of Corollary \ref{Cor:4.4} for $h := g - \bar g$, we find
\be \label{E:4.17}
\begin{split}
&\int_0^T\left(\nu \bar{g}(t) + \bar \mu +\kappa \iint_Q \bar \varphi(x,s) \bar{u}(x,s-t) \, dxds \right)
 (g(t)-\bar{g}(t)) \, dt  \geq  0
\end{split}
\ee
 for all $0 \le  g \le \beta$.
\begin{remark} For a Lagrange multiplier rule to hold, a regularity condition must be fulfilled. Here, the constraints
are obviously regular at any $\bar g$: Define $F: L^2(0,T) \to \mathbb{R}$ by 
\[
F(g) = \int_0^T g(\tau)\,d\tau - 1.
\]
Then 
\[
F'(g) h = \int_0^T h(\tau)\,d\tau,
\]
and hence $F'(g): L^2(0,T) \to \mathbb{R}$ is surjective for all $g \in L^2(0,T)$.
\end{remark}

A simple pointwise discussion of \eqref{E:4.17} leads to the following complementarity conditions for almost all
$t \in [0,T]$:
\begin{equation} \label{E:4.18}
 g(t) =\left\{
\begin{array}{rcl}
 0 &\mbox{if}&\ds  \nu \bar{g}(t) +  \bar \mu +\kappa\,
\iint_Q \bar \varphi(x,s) \bar{u}(x,s-t) \, dxds > 0\\[2ex]
\beta &\mbox{if}&\ds \nu \bar{g}(t) +  \bar \mu +\kappa\,
\iint_Q \bar \varphi(x,s) \bar{u}(x,s-t) \, dxds < 0.
\end{array}
\right.
\end{equation}
If the expression in right-hand side above vanishes, then we obviously have
\[
 g(t) = -\frac{1}{\nu} \left( \bar \mu +\kappa\,
\iint_Q \bar \varphi(x,s) \bar{u}(x,s-t) \, dxds\right).
\]
In a known way, the last three relations can be equivalently expressed by the projection formula
\[
\bar g(t) = \mathbb{P}_{[0,\beta]}\left(-\frac{1}{\nu} \left(\bar \mu +\kappa\,
\iint_Q \bar \varphi(s) \bar{u}(x,s-t) \, dxds \right) \right),
\]
where $\mathbb{P}_{[0,\beta]}: \mathbb{R} \to [0,\beta]$ is defined by
\[
\mathbb{P}_{[0,\beta]}(x) = \max(0,\min(\beta,x)).
\]

\section{Discussion of (PS)}

Let us now discuss the changes that are needed to establish the necessary optimality conditions
for the problem (PS) with the particular form \eqref{E:2.5} of $g$. Now,  $\kappa, t_1,$ and $t_2$
are our control variables. Let us denote by
$u_{(\kappa,t_1,t_2)}$ the unique state associated with $(\kappa,t_1,t_2)$.


The existence of  the derivatives $\partial_{t_i} u_{(\kappa,t_1,t_2)}, \, i=1,2$,  and
$\partial_{\kappa} u_{(\kappa,t_1,t_2)}$ can be shown
again by the implicit function theorem. We omit these details, because one can proceed analogously to the discussion 
for (PG). To shorten the notation, we write
\[
 z_i:=\partial_{t_i}  u_{(\kappa,t_1,t_2)}, \; i=1,2, \quad z_3 := \partial_{\kappa}  u_{(\kappa,t_1,t_2)}.
\]
By implicit differentiation, we find the functions $z_i$ from  linearized equations. Assume that the
derivatives have to be determined at the point $(\kappa,t_1,t_2)$ and fix the associated state \
$u:= u_{(\kappa,t_1,t_2)}$
for a while.
Then, $z_1$ solves
\be \label{E:5.1}
\begin{split}
 &\left(\partial_t z_1- \Delta z_1 +R'(u)z_1+\kappa\,  z_1\right)(x,t)= \left.\frac{\partial}{\partial t_1} \left[ \frac{\kappa\,}{t_2-t_1}
 \int_{t_1}^{t_2} u(x,t-\tau) \, d\tau \right]\right.\\
&\quad =\frac{\kappa}{t_2-t_1} \left[\frac{1}{t_2-t_1}\int_{{t_1}}^{{t_2}}
{u}(x,t-\tau) \, d\tau -  {u}(x,t-{t_1})
+ \int_{{t_1}}^{{t_2}}
z_1(x,t-\tau) \, d\tau\right].
\end{split}
\ee
Analogously, we find  for $z_2$
\be \label{E:5.2}
\begin{split}
 &(\partial_t  z_2- \Delta z_2 +R'(u)z_2+\kappa\,  z_2)(x,t)= \\
&\quad=\frac{- \kappa}{{t_2}-{t_1}} \left[\frac{1}{{t_2}-{t_1}}\int_{{t_1}}^{{t_2}}
{u}(x,t-\tau) \, d\tau -  {u}(x,t-{t_2})
+ \int_{{t_1}}^{{t_2}}
z_2(x,t-\tau) \, d\tau\right]
\end{split}
\ee
and for $z_3$
\[
\begin{split}
 &(\partial_t  z_3- \Delta z_3 +R'(u)z_3+\kappa \,   z_3 + u)(x,t) =
 \left.\frac{\partial}{\partial \kappa} \left[ \frac{\kappa\,}{t_2-t_1}
 \int_{t_1}^{t_2}u(x,t-\tau) \, d\tau \right]\right.\\
&\hspace{0.8cm}=\frac{1}{{t_2}-{t_1}} \left[ \int_{{t_1}}^{{t_2}}
{u}(x,t-\tau) \, d\tau  + \kappa \, \int_{{t_1}}^{{t_2}}
z_3(x,t-\tau) \, d\tau\right].
\end{split}
\]
Therefore, the equation for $z_3$ is
\begin{equation}  \label{E:5.3}
\begin{split}
&(\partial_t z_3- \Delta z_3 +R'(u)z_3+\kappa \,   z_3 )(x,t)- \frac{\kappa}{t_2-t_1} \int_{{t_1}}^{{t_2}}
z_3(x,t-\tau) \, d\tau \\
&\hspace{1cm} = \frac{1}{{t_2}-{t_1}} \int_{{t_1}}^{{t_2}}
{u}(x,t-\tau) \, d\tau  - u(x,t).
\end{split}
\end{equation}
Again, we introduce an adjoint equation to set up the optimality conditions. To this aim, let $(\kappa,t_1,t_2)$ 
an arbitrary fixed triplet and $u_{(\kappa,t_1,t_2)}$ be the associated state function. The adjoint equation
is
\be \label{E:5.5}
 \begin{array}{rcl}
 \left(-\varphi_t - \Delta \varphi +R'({u_{(\kappa,t_1,t_2)}})\varphi+ \kappa \, \varphi\right)(x,t) &=& \ds 
 \frac{\kappa}{ t_2- t_1}\int_{ t_1}^{ t_2} \varphi(x,t+\tau)
  \, d\tau\\[3ex]
  && + c_Q(x,t)({u}_{(\kappa,t_1,t_2)}(x,t) -u_d(x,t)) \\[1ex]
  &&\hspace{4cm} \text{ in } Q,\\[1ex]
 \partial_n \varphi&=& 0 \quad \text{in }\Sigma,\\[1ex]
  \varphi(x,t)&=&0 \quad \text{in } \Omega\times[T,2T]. \\
\end{array} 
\ee
This equation has a unique solution $\varphi \in L^\infty(Q)$ denoted by $\varphi_{(\kappa,t_1,t_2)}$ to indicate
the correspondence with $(\kappa,t_1,t_2)$. Existence and uniqueness can be shown in a standard way by the
substitution $\tilde t := T-t$ that transforms this equation to a standard forward equation that can be handled in
the same way as the state equation.

\begin{theorem}[Derivative of $J_S$] \label{T:5.1} Let $(\kappa,t_1,t_2)$ be given, $u:=u_{(\kappa,t_1,t_2)}$
be the associated state, and $\varphi := \varphi_{(\kappa,t_1,t_2)}$ be the associated adjoint state, i.e. the unique solution
of the adjoint equation \eqref{E:5.5}. Write for short $\delta := 1/(t_2-t_1)$. Then the partial derivatives of $J_S$ at $(\kappa,t_1,t_2)$ are given
by 
\begin{eqnarray*}
\partial_{t_1} J_S&=& 
 \nu t_1+  \frac{\kappa}{\delta} \iint_Q {\varphi}(x,t)
 \left[ \frac{1}{\delta}  \int_{{t_1}}^{{t_2}} {u}(x,t-\tau) \, d\tau -
  {u}(x,t- {t_1})\right] \, dxdt, \\[2ex]
\partial_{t_2} J_S&=&  \nu t_2-  \frac{\kappa}{\delta} \iint_Q {\varphi}(x,t)
 \left[ \frac{1}{\delta}  \int_{t_1}^{t_2} {u}(x,t-\tau) \, d\tau -
 {u}(x,t- {t_2}) \right]\, dxdt\\[2ex]
 \partial_{\kappa} J_S &=&  \nu \kappa + \iint_Q  \varphi(x,t) \left[ \frac{1}{\delta}  \int_{t_1}^{t_2}{u}(x,t-\tau)\, d\tau - u(x,t)\right]\, dxdt.
 \end{eqnarray*}
\end{theorem}
\begin{proof} We verify the expression for $\partial_{t_1} J_S(\kappa,t_1,t_2)$, the other formulas can be shown analogously.
To this aim, let $z_1 = \partial_{t_1}u_{(\kappa,t_1,t_2)}$ be the solution of the linearized equation \eqref{E:5.1}. For convenience,
we write $z:= z_1$ within this proof. Following the proof of Lemma \ref{L:4.3}, we multiply \eqref{E:5.1} by $\varphi$ and integrate
over $Q$. We obtain
\begin{equation}\label{E:5.6}
\begin{split}
&\iint_Q \left(\partial_t z \, \varphi + \nabla z \cdot \nabla \varphi + (R'(u) + \kappa)z\, \varphi \right)dxdt
\\[2ex]
&\qquad =\frac{\kappa}{\delta} \iint_Q \varphi(x,t) \left[\frac{1}{\delta} \int_{{t_1}}^{{t_2}}
{u}(x,t-\tau) \, d\tau- {u}(x,t-{t_1})\right]\, dxdt\\[1ex]
&\qquad \qquad +  \frac{\kappa}{\delta} \iint_Q \varphi(x,t) \int_{{t_1}}^{{t_2}} z(x,t-\tau) \, d\tau dxdt.
\end{split}
\end{equation}
Next, we multiply the adjoint equation \eqref{E:5.5} by $z$ and integrate over $Q$ to find
\begin{equation} \label{E:5.7}
\begin{split}
&\iint_Q \left(- z \, \partial_t \varphi + \nabla z \cdot \nabla \varphi + (R'(u) + \kappa)z\, \varphi \right)dxdt
\\[2ex]
&\qquad =  \frac{\kappa}{\delta}\iint_Q \int_{t_1}^{t_2} \varphi(x,t+\tau)\, d\tau \ z(x,t)\,   dxdt
+ \iint_Q c_Q(\bar u - u_d)\, z\, dxdt.
\end{split}
\end{equation}
Now recall that 
\[
 g(t) = \left\{
 \begin{array}{cl} \ds \frac{1}{t_2 -t_1}, & t_1 \leq t \leq t_2 \\[2ex]
 0, & \mbox{elsewhere}.
 \end{array} \right.
\]
Therefore, we can write
\[
\begin{split}
&\frac{1}{\delta} \iint_Q \varphi(x,t) \int_{{t_1}}^{{t_2}} z(x,t-\tau) \, d\tau dxdt 
=  \iint_Q \int_0^T \varphi(x,t) g(\tau) z(x,t-\tau) \, d\tau dxdt \\
&= \iint_Q \int_{0}^{T} g(\tau) \varphi(x,t+\tau)d\tau z(x,t)dxdt = 
\frac{1}{\delta}\iint_Q \int_{t_1}^{t_2} \varphi(x,t+\tau)d\tau z(x,t)dxdt,
\end{split}
\]
where the second equation follows from \eqref{E:Fubini}. In view of 
\[
\iint_Q (- z) \, \partial_t \varphi \, dxdt = \iint_Q \varphi \, \partial_t z \,  dxdt,
\]
a comparison of the equations \eqref{E:5.6} and \eqref{E:5.7} yields
\[
\begin{split}
&
\iint_Q c_Q(u - u_d)\, z\, dxdt = \\
&
\qquad \qquad = \frac{\kappa}{\delta}\iint_Q  \varphi(x,t) \left[\frac{1}{\delta}\int_{{t_1}}^{{t_2}}
{u}(x,t-\tau) \, d\tau- {u}(x,t-{t_1})\right]dxdt.
\end{split}
\]
Since 
\[
\partial_{t_1} \left[\frac{1}{2}\iint_Q c_Q (u_{(\kappa,t_1,t_2)}-u_d)^2 dxdt \right]= \iint_Q c_Q (u_{(\kappa,t_1,t_2)}-u_d) \, z_1\, dxdt,
\]
the first claim of the theorem follows immediately.
\hfill $\Box$ \end{proof}

As a direct consequence of the theorem on the derivative of $J_S$, we obtain the following corollary.
\begin{corollary}[Necessary optimality condition for (PS)] Let $(\bar \kappa,\bar t_1,\bar t_2)$ be optimal for the problem
{\rm (PS)} and let $\bar u := u_{(\bar \kappa,\bar t_1,\bar t_2)}$ and $\bar \varphi := \varphi_{(\bar \kappa,\bar t_1,\bar t_2)}$
denote the associated state and adjoint state, respectively. Then, with the gradient 
$\nabla J_S(\bar \kappa,\bar t_1,\bar t_2)$ defined by Theorem  \ref{T:5.1} with $\bar \varphi$ and $\bar u$
inserted for $\varphi$ and $u$, respectively, the variational inequality
\begin{equation}
 \nabla J_S(\bar \kappa,\bar t_1,\bar t_2)\cdot (\kappa - \bar \kappa ,t_1 - \bar t_1,t_2 - \bar t_2)^\top \ge 0
 \quad \forall (\kappa, t_1, t_2) \in C_\delta
 \end{equation}\label{E:5.6b}
 is satisfied.
 \end{corollary}

Since the variable $\kappa$ is unrestricted, the associated part of the variational inequality amounts to
\[
\nu \bar \kappa - \iint_Q  \bar \varphi(x,t)  \bar u(x,t)\, dxdt
  +  \frac{1}{{\bar t_2}-{\bar t_1}}  \iint_Q {\bar \varphi}(x,t) {\bar u}(x,t-\bar t_2)\,dxdt = 0.
\]
If $(\bar t_1,\bar t_2)$ belongs to the interior of the admissible set $C_\delta$, then  the associated components of
$\nabla J_S$ must vanish as well, hence
\begin{eqnarray*}
\nu \bar t_1+  \frac{\bar \kappa}{{\bar t_2}-{\bar t_1}} \iint_Q {\bar \varphi}(x,t)\left[
 \frac{1}{{\bar t_2}-{\bar t_1}}\int_{{\bar t_1}}^{{\bar t_2}} {\bar u}(x,t-\tau) \, d\tau   -   {\bar u}(x,t- {\bar t_1}) \right]\, dxdt
& =& 0\\
 \nu \bar t_2-  \frac{\bar \kappa}{{\bar t_2}-{\bar t_1}} \iint_Q {\bar \varphi}(x,t) \left[\frac{1}{{\bar t_2}-{\bar t_1}}
\int_{{\bar t_2}}^{{\bar t_2}} {\bar u}(x,t-\tau) \, d\tau-
 {\bar u}(x,t- {\bar t_2})\right]\, dx dt &=& 0.
\end{eqnarray*}

\begin{remark}[Application of the formal Lagrange technique] \label{Rem:4.5} To find the form of $\nabla J_S$ and for 
establishing the necessary optimality conditions, it might be easier to apply the following standard technique that is slightly 
formal but leads to the same result:
We set up the Lagrangian function $\mathcal{L}$,
\[
\begin{split}
 &\mathcal{L}(u,\kappa,t_1,t_2,\varphi)= \frac{1}{2}\iint_Q c_Q (u-u_d)^2 dxdt -\iint_Q \varphi(x,t) \, \cdot \\
 &\qquad \cdot \left[(\partial_t u - \Delta u +R(u)
 +\kappa\,u)(x,t) -  \frac{\kappa}{t_2-t_1} \int_{t_1}^{t_2}u(x,t-\tau) \, d\tau  \right] dxdt.
 \end{split}
\]
If $(\kappa,t_1,t_2)$ is a given triplet, then the adjoint equation for the adjoint state $\varphi_{(\kappa,t_1,t_2)}$
is found by
\[
\partial_u \mathcal{L}(u,\kappa,t_1,t_2,\varphi)\, v =0 \quad \forall  v \, : \, v(\cdot,t)=0 \mbox{ for } t \leq 0.
\]
The derivatives of $J_S$ are obtained by associated derivatives of $\mathcal{L}$. For instance, we have
\[
\partial_{t_i}  J_S(\kappa,t_1,t_2) = \partial_{t_i} \mathcal{L}(u,\kappa,t_1,t_2,\varphi)
\]
if $u = u_{(\kappa,t_1,t_2)}$ and $\varphi = \varphi_{(\kappa,t_1,t_2)}$ are inserted after having taken
the derivative of $\mathcal{L}$  with respect to $t_i$. This obviously yields the first two components of $\nabla J_S$
in Theorem \ref{T:5.1}. For the third component with proceed analogously. 
\end{remark}

\section{Numerical examples for (PS)}

\subsection{Introductory remarks}

The numerical solution of the problems posed above requires techniques that are adapted to the desired type
of patterns $u_d$. In this paper, we concentrate on numerical examples for the simplified problem (PS), where the kernel $g$ is
a step function. Although this
problem is mathematically equivalent to a nonlinear optimization problem in a convex admissible set of $\mathbb{R}^3$,
the obtained patterns are fairly rich and interesting in their own. In a forthcoming paper to be published elsewhere,
we will report  on the numerical treatment of the more general problem (P), where a kernel function $g$ is to be
determined.

In this section, we present some results for the problem (PS), where a standard regularized tracking type functional $J$
is to be minimized in the set $C_\delta$.
It will turn out that (PS) is only suitable for tracking desired states $u_d$ of simple
structure. In all what follows, $\Omega = (a,b)$ is an open interval, i.e. we concentrate on the spatially one-dimensional
case.

Compared to many optimal control problems for semilinear parabolic equations that were investigated in literature,
the numerical solution of the problems posed here is a bit delicate. We are interested in approximating desired states $u_d$
that exhibit certain geometrical patterns. If they have a periodic structure, then the objective function
$J$ may exhibit many local minima with very narrow regions of attraction for the convergence of numerical techniques.
Therefore, the optimization methods should be started in a sufficiently small neighborhood around the
desired optimal solution. Moreover, the standard functional $J$ does not really fit to our needs. We will address
the tracking of periodic patterns $u_d$ in Section \ref{S7}.

\subsection{Discretization of the feedback system}
To discretize the feedback equation \eqref{E:2.5}, we apply an implicit Euler scheme with respect to the time
and a finite element approximation by standard piecewise linear and continuous ansatz functions (''hat functions'')
with respect to  the space variable.

For this purpose, the define a time grid by an equidistant partition of $[0,T]$ with mesh size $\tau = T/m$ and node points
$t_j = j \, \tau$, $j = 0,\ldots,m$. Associated with this time grid, functions $u_j: \Omega \to \mathbb{R}$ are to be computed
that approximate $u(\cdot,t_j)$, $j = 0,\ldots,m$, i.e. $u_j \sim u(\cdot,t_j)$. Based on the functions $u_j$, we
define grid functions $u^\tau: Q \to \mathbb{R}$ by piecewise linear approximation,
\[
u^\tau(x,t) = \frac{1}{\tau}[(t_{j+1}-t)\, u_j(x) + (t_{j}-t)\, u_{j+1}(x)], \mbox{ if } t \in [t_j,t_{j+1}], \ j = 0,\ldots,m.
\]
By the implicit Euler scheme, the following system of
nonlinear equations is set up,
\[
\begin{split}
 & \frac{u^\tau(x,t_{j+1})-u^\tau(x,t_{j})}{\tau} - \Delta u^\tau(x,t_{j+1})+R(u^\tau(x,t_{j+1})) \\
 &\hspace{1cm}= \frac{\kappa}{t_2-t_1} \int_{t_1}^{t_2} u^\tau(x,t_{j+1}-s)\: ds - \kappa \, u^\tau(x,t_{j+1}), \quad j = 0,\ldots,m-1.
\end{split}
\]
The spatial approximation is based on an equidistant partition of $\Omega=(a,b)$ with mesh size $h>0$. Here, we define standard
piecewise affine and continuous ansatz functions (hat functions) $\phi_i: \Omega \to \mathbb{R}$ and approximate the grid
function $u^\tau$ by $u_h^\tau$,
\[
u_h^\tau(x,t_j)=\sum_{i=0}^n u_{ji}\phi_i(x)
\]
with unknown coefficients $u_{ji} \in \mathbb{R}$, $j = 0,\ldots,m, \, i = 0,\ldots,n$.

To set up the discrete system, we define vectors $\mathbf{u}^j \in \mathbb{R}^{n+1}$ by
$\mathbf{u}^j = (u_{j0},\ldots,u_{jn})^\top$, $j = 0,\ldots,m$. Moreover, we establish the mass and stiffness matrices
\[
M:= \left(\int_\Omega \phi_k(x) \phi_\ell(x) \: dx\right)_{k,\ell=0}^n, \quad
 A:=\left(\int_\Omega \phi_k'(x) \phi_\ell'(x) \: dx\right)_{k,\ell=0}^n\, ,
\]
and, for $j = 0,\ldots,m$, the vectors
\begin{eqnarray*}
 && R(\mathbf{u}^j):=\left( \int_\Omega R(u_h^\tau(x,t_j))\, \phi_i(x)\: dx\right)_{i=0}^n\\
 && K(\mathbf{u}^j):=\frac{\kappa}{t_2-t_1}\left(\int_\Omega \int_{t_1}^{t_2} (u_h^\tau(x,t_j)-s)\,ds \, \phi_i(x)  \: dx\right)_{i=0}^n.
\end{eqnarray*}
\begin{remark}To compute the integrals $R(\mathbf{u}^j)$, we invoke a 4 point Gauss integration. 
Notice that the functions $x \mapsto R(u_h^\tau(x,t_j))$
are third order polynomials so that the integrand of $R(\mathbf{u}^j)$ is  a polynomial of order 4. 
Here, the 4 point Gauss integration is exact.
\end{remark}

For the computation of the vectors $K(\mathbf{u}^j)$, we use the trapezoidal rule. Here,  for $t_j-s \leq 0$, the values $u_0(x,t_j-s)$ must be inserted.
To increase the precision, the primitive of $u_0$ is used. The complete discrete system is
\[
 M\mathbf{u}^{j+1} - M\mathbf{u}^j+\tau A\mathbf{u}^{j+1} +{\tau}\, R(\mathbf{u}^{j+1}) =\tau K(\mathbf{u}^{j+1})-\tau \kappa M\mathbf{u}^{j+1}
\quad j=0,1,\ldots m-1.
\]
 We define
\[
 F(\mathbf{u}):=[(1+\tau \kappa)M+\tau A]\mathbf{u}+{\tau}\, R({\mathbf{u}})-\tau K(\mathbf{u}) .
 \]
In each time step, we solve the nonlinear equation  $F(\mathbf{u})=M \mathbf{u}^{j}$ to obtain $\mathbf{u}^{j+1}$.
To this aim, we apply a fixed point iteration. 

\subsection{Numerical examples for the standard tracking functional}

For testing our numerical method, we generate the desired state $u_d$ as solution of the feedback system for a given triplet $(\hat \kappa,\hat t_1,\hat t_2)$, i.e.
\[
u_d := u_{(\hat \kappa,\hat t_1,\hat t_2)}.
\]
If the regularization parameter $\nu$ is small, then the numerical solution of (PS) should return a vector that is close
to $(\hat \kappa,\hat t_1,\hat t_2)$.

In all of our computational examples,  we selected a small number $\delta > 0$ so that the restriction $t_2 -t_1\ge \delta$ was
never active. Moreover, the Tikhonov regularization parameter $\nu$ was set to zero, because a regularization was not needed
for having numerical stability.

In the first numerical examples of this paragraph, the aim is to generate desired wave type solutions  that expand with a given velocity.
\vspace{1ex}

\begin{example}[Desired traveling wave with pre-computed $u_d$] \label{Ex1}  We select $\Omega = (-20,20)$, $T = 40$,
$u_1 = 0, \, u_2 = 0.25, u_3 = 1$, $\kappa = 0.5$. Moreover, we take as initial function
\[
u_0(x,t) := \frac{1}{2}\ds\left(1-\tanh\left(\frac{x-vt}{2\sqrt{2}}\right)\right), \,
 \quad x \in \Omega,\, t \le 0,
\]
where $v=(1-2u_2)/\sqrt{2}$ is the velocity of the uncontrolled traveling wave given by $u_0$.
Following the strategy explained above, we fix the triplet 
$(\hat \kappa,\hat t_1,\hat t_2)$ by $(0.5,0.456,0.541)$ and  obtain $u_d = u_{(\hat \kappa,\hat t_1,\hat t_2)}.$
This is a traveling wave with a smaller velocity $v_d \approx 0.25$ due to the control term. To test our 
method, we apply our optimization algorithm to find $(\bar \kappa,\bar t_1,\bar t_2)$ such that the 
associated state function $u$ coincides with $u_d$.
The method should return a result  $(\kappa, t_1, t_2)$ that is close to 
the vector $(0.5,0.456,0.541)$.
\end{example}

To solve the problem (PS) in Example \ref{Ex1}, we applied the \textsc{Matlab} code \texttt{fmincon}.
 For the gradient $\nabla J_S(\kappa,t_1,t_2)$,  a subroutine was implemented by our adjoint calculus.
In this way, we were able to use the differentiable mode of \texttt{fmincon}.

During the optimization, we fixed $\kappa = 0.5$  and considered the minimization only with respect
to $(t_1,t_2)$. Taking $t_1 =0,\, t_2 = 1$ as initial iterate for the optimization, \texttt{fmincon} returned $t_1 = 0.456, \, t_2 = 0.541$
as solution; notice that we fixed $\nu = 0$. The result is displayed in Fig.\,\ref{fig:Ex1}.

\begin{figure}[h]\centering
  \includegraphics[scale=0.23]{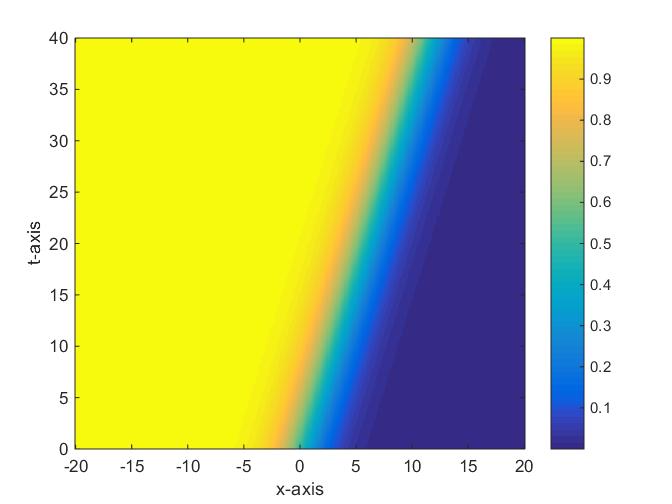}\hspace{20pt}
  \includegraphics[scale=0.23]{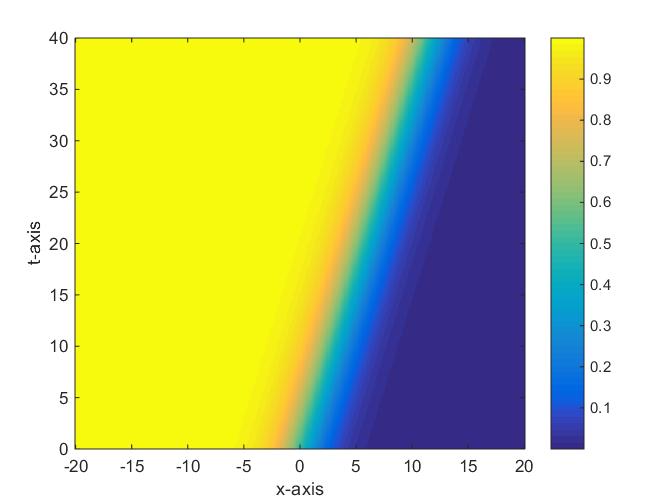}
  \caption{Example \ref{Ex1}, desired traveling wave $u_d$ (left) and computed optimal state  $u$ (right).}\label{fig:Ex1}
\end{figure}

\begin{example}[Stopping a traveling wave ]\label{Ex2} In contrast to Example \ref{Ex1}, here we fix the desired pattern $u_d$
that is displayed in Fig.\,\ref{fig:Ex2}, left side. This desired pattern was not pre-computed but geometrically designed, i.e.
it is a ''synthetic'' pattern that shows a traveling wave stopping at the time $t \approx 16$.
The other data are selected as in Example \ref{Ex1}.
\end{example}

In the optimization process for Example \ref{Ex2}, we fixed $\kappa = -1.5$. The initial iterate was $t_1 = 0, \, t_2 = 1$; \texttt{fmincon} returned
$t_1 = 0.05, \, t_2 = 0.94$ as solution. The optimal state is displayed in Fig.\,\ref{fig:Ex2}.

\begin{figure}[h]\centering
  \includegraphics[scale=0.23]{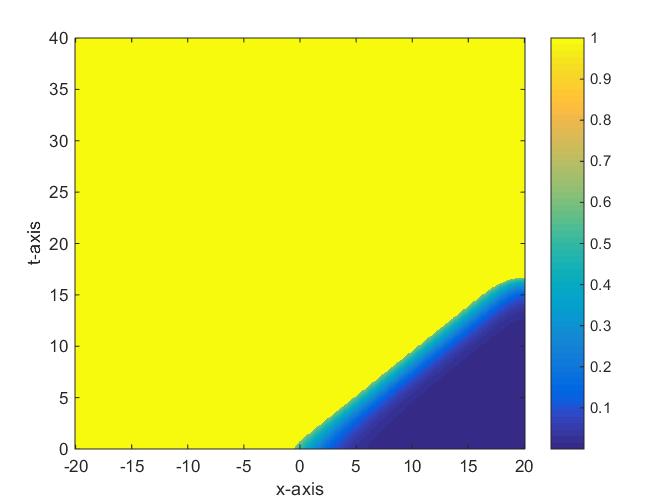}\hspace{20pt}
  \includegraphics[scale=0.23]{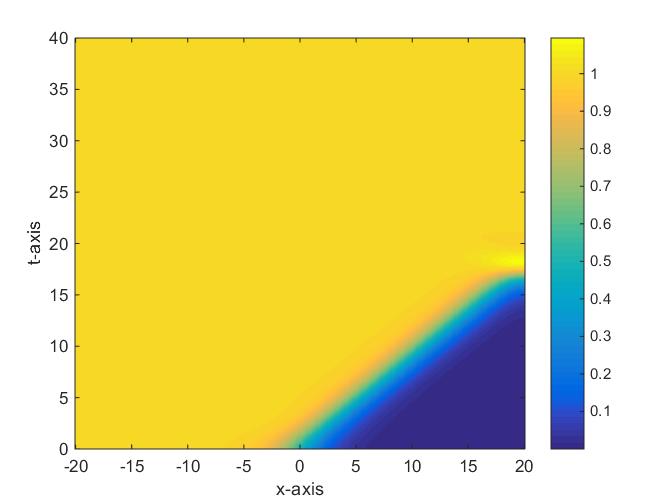}
  \caption{Example \ref{Ex2}, synthetic desired state $u_d$ (left) and computed optimal state  $u$ (right).}\label{fig:Ex2}
\end{figure}

The next example shows that the applicability of the standard tracking functional $J$ of (PS) is limited to simple
patterns $u_d$, e.g. wave  type solutions of constant velocity.

\begin{example}[Periodic pattern] \label{Ex3} Also here, $u_d$ is a synthetic pattern that was not precomputed.
In $\Omega = (-20,20)$, we define
\[
u_d(x,t) =  3 \, \sin(t - \cos(\frac{\pi}{40}(x+40))).
\]
Notice that this function $u_d$ obeys the homogeneous Neumann boundary conditions. It  is displayed in
Fig.\,\ref{fig:Ex3}, left. Since such a periodic pattern cannot be expected for small times, we consider the
tracking only on $[T/2,T]$. Therefore,  here we re-define the objective functional $J$ by
\[
J_S(\kappa,t_1,t_2) := \int_{T/2}^T\int_\Omega (u_{(\kappa,t_1,t_2)} - u_d)^2 \, dxdt + \frac{\nu}{2} (\kappa^2 + t_1^2+t_2^2).
\]
\end{example}
During the optimization run, we fixed the values $\kappa = -2$ and $t_1 = 0$ and optimized only with respect to $t_2$. Starting from
$t_2 = 2$, the code \texttt{fmincon} returned the solution $t_2 = 3.71$. At this point, the Euclidean norm  of
$\nabla J_S$ is $|\nabla J_S(-2,0,3.71)| \approx 0.045$. This is a fairly good value and indicates that the result
should be close to a local minimum. Nevertheless, the computed optimal objective value is very large,
\[
J_S(-2,0,3.71) = 1.32\cdot 10^3.
\]
\begin{remark}In this and in the next examples, we fix $t_1 = 0$. We observed in our computational examples that the optimization
with respect to $(\kappa,t_2)$ and $t_1 = 0$ yields sufficiently good results. Moreover, we found examples, where we got the same optimal
objective value of $J$ for very different triplets $(\kappa,t_1,t_2)$.
\end{remark}

\begin{figure}[h]\centering
  \includegraphics[scale=0.28]{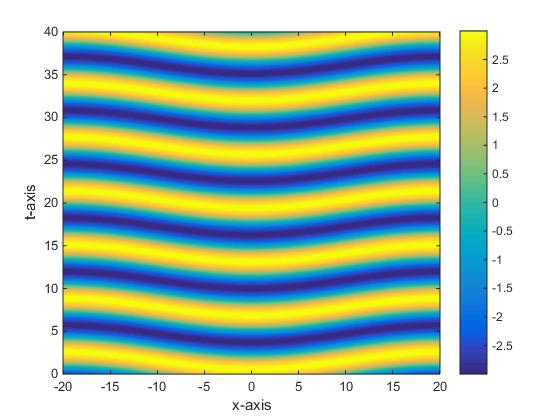}\hspace{20pt}
  \includegraphics[scale=0.23]{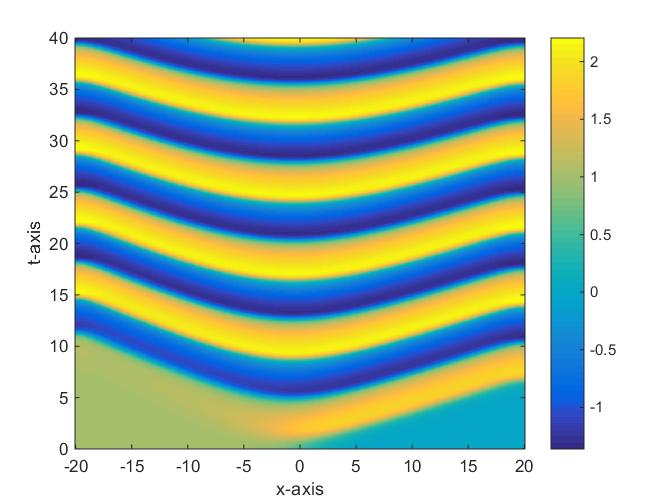}
  \caption{Example \ref{Ex3}, Desired periodic pattern $u_d$ (left) and computed optimal state  $u$ (right).}\label{fig:Ex3}
\end{figure}
The computed optimal state for Example \ref{Ex3} is far from the desired one. In particular,
the temporal periods are very different. The reason is that the standard quadratic tracking type functional $J$
is not able to resolve the desired periodicity. The main point is that the $L^2$-norm of the difference of a time-periodic function
$t \mapsto u_d(t)$ and its phase shifted function $t \mapsto u_d(t-s)$ can be very large, although both functions have
the same time period. For instance, in $Q = (-20,20)\times (0,40)$ we have 
\[
\iint_Q \left(3 \sin(t - \cos(\frac{\pi}{20}(x+20))) - 3 \sin(t - 3 - \cos(\frac{\pi}{20}(x+20)))\right)^2 dxdt =  I,
\]
where $I  \approx 1.4374\cdot 10^4$.
This brought us to considering another type of objective functionals that is discussed in the next section.

\section{Minimizing a cross correlation type functional}\label{S7}

\subsubsection*{The cross correlation}

As Example \ref{Ex3} showed,
we need an objective functional that better expresses our goal of approximating periodic structures.
This is the {\em cross correlation} between $u$ and $u_d$. Moreover, in the functional, we have to
observe a later part  of the time interval, where $u$  already has developed its periodicity.

Let us briefly explain the meaning of the cross correlation. Assume that $x_d: \mathbb{R} \to \mathbb{R}$ is
a periodic function and $x: [0,T] \to \mathbb{R}$ is another function; think of a function $x$ that is
identical with  $x_d$ after a time shift.

To see, if $x_d$ and $x$ are time shifts of each other, we consider the extremal problem
\begin{equation} \label{E:min_s}
\min_{s \in \mathbb{R}} \int_0^T (x(t) - x_d(t-s))^2 dt.
\end{equation}
If $x$ and $x_d$ are just shifted, then the minimal value in \eqref{E:min_s} should be zero by taking the correct time shift $s$.
The functional  \eqref{E:min_s} can be simplified. To this aim, we expand the integrand,
\[
\int_0^T (x(t) - x_d(t-s))^2 dt  = \int_0^T x^2(t)\, dt - 2 \int_0^T x(t)x_d(t-s) dt + \int_0^T x_d^2(t-s)\, dt.
\]
The first integral does not depend on  $s$. Since $x_d$ is a periodic function, also the last integral
is independent of the shift $s$. Therefore, instead of minimizing the quadratic functional above, we can solve
the following problem:
\begin{equation} \label{E:7.2}
\max_{s \in \mathbb{R}}  \frac{\int_0^T x(t)x_d(t-s) dt}{\|x\|_{L^2(0,T)} \|x_d\|_{L^2(0,T)}},
\end{equation}
where we additionally introduced a normalization.
The result is the so-called {\em cross correlation} between $x$ and $x_d$. The largest possible value in \eqref{E:7.2} is 1; in this case,
both functions are collinear.

In the application to our control problems, this induces two equivalent objective functionals.
The minimization problem \eqref{E:min_s} leads to the optimization problem
\begin{equation} \label{E:min_s_b}
\min_{(\kappa,t_1,t_2) \in C_\delta} \left( \, \min_{s \in \mathbb{R}}
\int_Q(u(x,t) - u_d(x,t-s))^2 dxdt+  \frac{\nu}{2}\, (\kappa^2 + t_1^2 + t_2^2)
\right).
\end{equation}

The other way around is an equivalent problem that uses the cross correlation,
\begin{equation} \label{E:min_s_corr}
\min_{(\kappa,t_1,t_2) \in C_\delta} J_{corr}(\kappa,t_1,t_2)
\end{equation}
where
\begin{equation}
J_{corr}(\kappa,t_1,t_2) :=  1 - \max_{s \in \mathbb{R}}
\frac{\ds \iint_Q  u(x,t)u_d(x,t-s) dxdt}{\ds \|u\|_{L^2(Q)} \|u_d\|_{L^2(Q)}}
+ \frac{\nu}{2}\, (\kappa^2 + t_1^2 + t_2^2).
\end{equation}

For solving \eqref{E:min_s_corr}, we applied the \textsc{Matlab} code  \texttt{pattern search} that turned out to be
fairly efficient in finding global minima for functions that exhibit multiple local minima.
In the periodic case, we have to deal with multiple local minima indeed.
Here, the cross correlation based functional $J_{corr}$ is more useful, as the next figure shows.

\begin{example}[Multiple local minima] \label{Ex4}
We pre-compute the desired state by $u_d = u_{(-2,0,2.5)}$ and consider the functions
$\kappa \mapsto J_S(\kappa,0,2.5)$ and $   \kappa \mapsto J_{corr}(\kappa,0,2.5)$ for $\kappa$ around the (optimal)
parameter $\bar \kappa = -2$.
\end{example}

The two functions defined in Example \ref{Ex4} are  shown in Fig.\,\ref{fig:Ex4}. The function  $\kappa \mapsto J_S(\kappa,0,2.5)$
behaves more wildly around $\kappa = -2$ than the function $\kappa \mapsto J_{corr}(\kappa,0,2.5)$ that
is based upon the cross correlation.

\begin{figure}[h]\centering
  \includegraphics[scale=0.23]{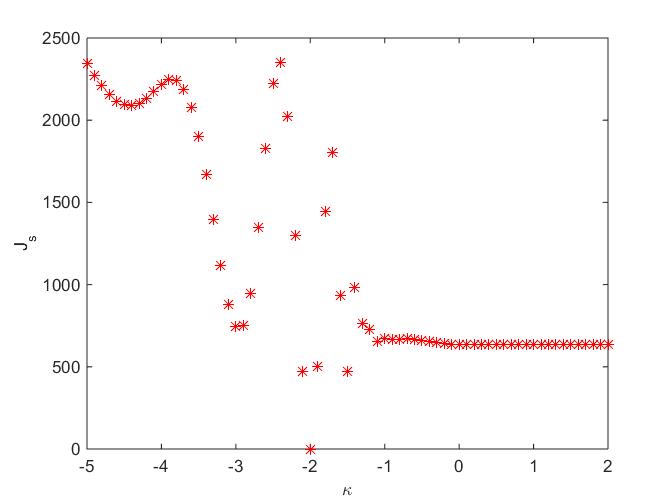}\hspace{20pt}
  \includegraphics[scale=0.23]{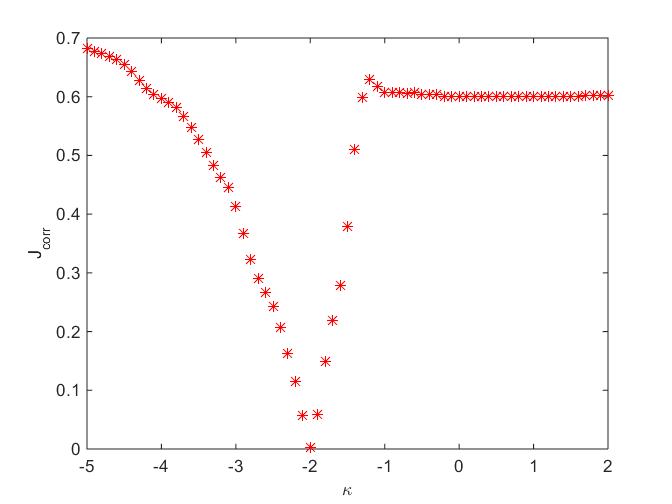}
  \caption{Example \ref{Ex4}, goal functions $\kappa \mapsto J_S(\kappa,0,2.5)$ (left) and $\kappa  \mapsto J_{corr}(\kappa,0,2.5)$ (right).}\label{fig:Ex4}
\end{figure}

Let us re-consider the optimization problem of Example \ref{Ex3}, but now by the cross correlation based
optimization problem \eqref{E:min_s_corr}. Here, we apply the following strategy:
We keep $t_1=0$ fixed and optimize only with respect to $(\kappa,t_2)$. Moreover, at the beginning we fixed
$\kappa = -2$ and optimized with respect to $\kappa$ in a second run. The computed solution was  $t_2 = 2.7631$
with a value $J_{corr}=0.1604, \, (t_1=0, \, \kappa =-2)$; as before the Tikhonov parameter was selected as $\nu = 0$.
The computed $u$ of this first step is shown in Fig.\,\ref{fig:Ex5}. Now the agreement, in particular of the temporal 
period, is much better. 

Next, we performed an alternating search for $(\kappa,t_2)$ starting with the result obtained in the first step.
We obtained $t_2 = 2.7631, \, \kappa = -2.4318$ and the improved objective value $J_{corr} = 0.1229$.
This improvement  is graphically hardly to distinct from Fig. \ref{fig:Ex5}.
\begin{figure}[h]\centering
  \includegraphics[scale=0.28]{ud.jpg}\hspace{20pt}
  \includegraphics[scale=0.28]{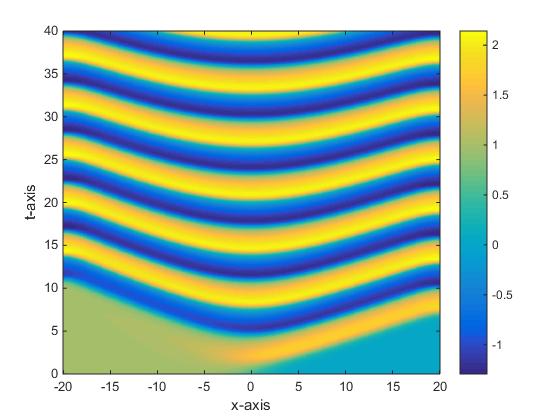}
  \caption{Example \ref{Ex3}, Desired pattern $u_d$ (left) and computed pattern after minimizing $J_{corr}$ with respect to $t_2$ (right)
  for fixed $\kappa = -2,\, t_1 = 0$.}\label{fig:Ex5}
\end{figure}
%
%
%


Finally, we consider another example with synthetic $u_d$ that has a larger period than $u_d$ of Example \ref{Ex3}.

\begin{example} \label{Ex7}We consider again $\Omega = (-20,20)$ and observe $u_d$ only in the time interval $[20,40]$.
For $u_d$ we select
\[
u_d(x,t) =  3 \, \ds\sin\left(\frac{t}{2} - \cos\left(\frac{\pi}{20}(x+20)\right)\right).
\]
Again, $t_1 = 0$ is fixed and the iteration is started with $t_2=2,\, \kappa=-2$. The optimal control parameters
are $t_2 = 6.94191, \, \kappa=-2.28$ with computed optimal objective value $J_{corr}=0.1209$. The computed 
optimal state is shown in Fig.\,\ref{fig:Ex7}.
\end{example}

\begin{figure}[h]\centering
  \includegraphics[scale=0.28]{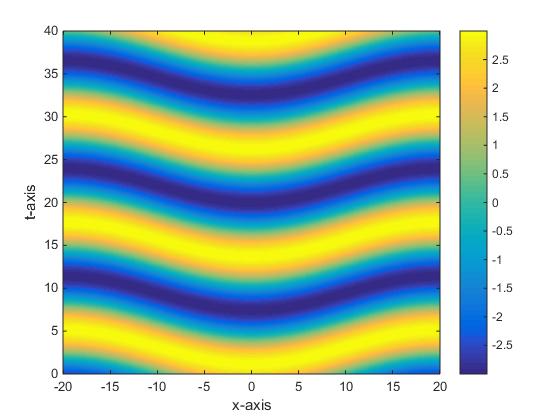}\hspace{20pt}
  \includegraphics[scale=0.28]{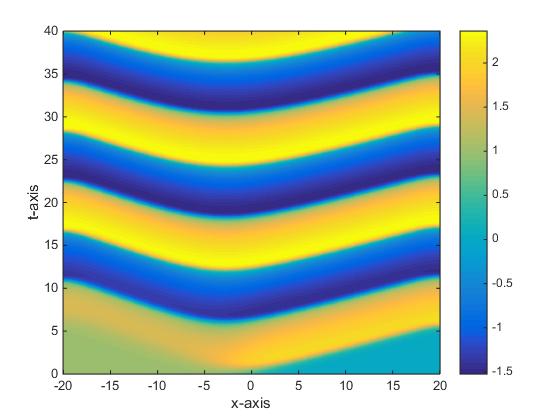}
  \caption{Example \ref{Ex7}, desired pattern $u_d$ (left) and optimal pattern (right)
  for fixed $t_1 = 0$.}\label{fig:Ex7}
\end{figure}


\end{document}